\documentclass[11pt]{article}

\title{Coloring in Graph Streams}

	\author{%
  Suman K. Bera
  \thanks{Department of Computer Science, Dartmouth College.}
  \and
  Prantar Ghosh $^\fnsymbol{footnote}$
}
\date{\today}

\usepackage[margin=1in]{geometry}
\usepackage{amsmath,amsthm,amssymb,xspace,verbatim,paralist,enumitem,multirow,bm}
\usepackage[usenames, dvipsnames]{xcolor}
\usepackage{times,txfonts,bm}
\usepackage[T1]{fontenc}
\usepackage{array,multirow}
\usepackage{algorithm,algorithmicx}
\usepackage[noend]{algpseudocode}
\usepackage{multicol}
\usepackage{graphicx}

\usepackage{tikz,relsize}
\usetikzlibrary{automata,positioning}

\setlist[enumerate]{
  noitemsep
}

\usepackage{url}
\PassOptionsToPackage{colorlinks, linkcolor=BrickRed, citecolor=blue}{hyperref}
\usepackage[colorlinks, linktoc=all, linkcolor=BrickRed, citecolor=blue]{hyperref}
\usepackage{cite,cleveref}

\allowdisplaybreaks[1]

\newtheorem{theorem}{Theorem}[section]
\newtheorem{lemma}[theorem]{Lemma}

\theoremstyle{definition}  \newtheorem{definition}[theorem]{Definition}

\newcommand{\EE}{\mathbb{E}}

\newcommand{\eps}{\varepsilon}
\newcommand{\tO}{\widetilde{O}}

\newcommand{\etal}{{et al.}\xspace}

\newcommand{\ceil}[1]{\lceil{#1}\rceil}

\renewcommand{\geq}{\geqslant}
\renewcommand{\leq}{\leqslant}

\DeclareMathOperator{\poly}{poly}
\DeclareMathOperator{\polylog}{polylog}

\DeclareMathOperator{\degen}{\kappa}
\DeclareMathOperator{\col}{\psi}
\DeclareMathOperator{\out}{out}

\algdef{SE}[SUBALG]{Indent}{Unindent}{}{\algorithmicend\ }%
\algtext*{Indent}
\algtext*{Unindent}

\begin{document}

\maketitle

\begin{abstract}
In this paper, we initiate the study of the vertex coloring
problem of a graph in the semi streaming model. In this model,
the input graph is defined by a stream of edges, arriving in 
adversarial order and any algorithm must process the edges in 
the order of arrival using space linear (up to polylogarithmic factors)
in the number of vertices of the graph. In the offline settings, there
is a simple greedy algorithm for $(\Delta+1)$-vertex coloring of a graph
with maximum degree $\Delta$. We design a one pass randomized
streaming algorithm for $(1+\eps)\Delta$-vertex coloring problem
for any constant $\eps >0$ using $O(\eps^{-1} n \poly\log n)$ space where
$n$ is the number of vertices in the graph. Much more color
efficient algorithms are known for graphs with bounded arboricity
in the offline settings. Specifically, there is a simple
$2\alpha$-vertex coloring algorithm for a graph with 
arboricity $\alpha$. We present a $O(\eps^{-1}\log n)$ pass randomized
vertex coloring algorithm that requires at most $(2+\eps)\alpha$
many colors for any constant $\eps>0$ 
for a graph with arboricity $\alpha$ in the semi
streaming model.
\end{abstract}

\section{Introduction}

Graph coloring is a fundamental and one of the most intensively 
studied problems in computer science:
given a graph $G=(V,E)$, assign colors to the vertices in $V$
such that no two adjacent vertices share the same color. 
The minimum number of colors needed to color a graph is 
called the chromatic number of the graph. Finding the chromatic
number of a graph is a notoriously hard problem: assuming 
$P \neq NP$, there does not exist any polynomial time algorithm
that approximates the chromatic number of a graph within a 
factor of $n^{1-\eps}$ for any constant $\eps >0$ where $n$ is 
the number of vertices in the graph~\cite{feige1996zero,zuckerman2006linear,khot2006better}. 
However, there are known upper bounds on the chromatic number. 
A graph with maximum degree $\Delta$ can be colored using
at most $\Delta+1$ many colors. There is a simple linear
time greedy coloring algorithm
\footnote{A greedy coloring algorithm considers the vertices of a 
graph sequentially, assigning it the first available color permitted
by the already colored neighbors of the vertex.}
that achieves such a coloring. A far superior bound is known for 
a rather large class of graph families: graphs with bounded degeneracy
\footnote{A graph $G$ is $\degen$-degenerate if every subgraph of $H$ of
$G$ has a vertex of degree at most $\degen$. Degeneracy of a graph is 
the minimum $\degen$ for which the graph is $\degen$-degenerate. It is
easy to see that $\degen \leq \Delta$ for any graph.}.
A graph with degeneracy $\degen$ has a chromatic number of at most $\degen + 1$.
A greedy coloring algorithm on the following vertex ordering 
produces such a coloring of the graph. Pick the vertex with 
minimum degree, remove the edges incident on this vertex,
recursively order the remaining vertices, then place 
this vertex at the end of the list. Arboricity
\footnote{Arboricity of an undirected graph is the 
minimum number of forests into which
the edges of the graph can be partitioned.}
of a graph is a closely related concept to degeneracy. 
For a graph with degeneracy $\degen$ and arboricity $\alpha$,
the following relations hold: $\alpha \leq \degen \leq 2\alpha -1$.
Hence, a graph with arboricity bounded by $\alpha$ has a chromatic 
number of at most $2\alpha$, and the above greedy coloring
algorithm produces such a coloring.

In this work, we initiate the study of the graph coloring problem
in the semi streaming model.
In this model, the input graph is presented as a stream of edges;
any algorithm  must process the edges in the order of its arrival, 
in one or more passes, using $O(n \polylog n)$ space, 
where $n$ is the number of vertices
~\cite{Muthukrishnan05,FeigenbaumKMSZ05}. 
The goal of the problem is to color the vertices of the graph using 
as few colors as possible. We first investigate whether 
we can find a $(1+\eps)\Delta$-vertex coloring of a 
graph with maximum degree $\Delta$ 
for any constant $\eps>0$ in this model. We show that such a coloring 
is possible to find in one pass. Our space usage involves a
factor of $\log n$ and an $\eps$ dependant factor;
we omit these factors for clarity, hiding them
into an $\tO(\cdot)$ notation.  
\begin{itemize}
    \item There is a randomized one pass streaming algorithm
    that finds a $(1+\eps)\Delta$-vertex coloring of the input
    graph using $\tO(n)$ amount of space.
\end{itemize}
We then explore whether one can significantly improve the 
number of colors for more general graph families. In this regard,
we consider bounded arboricity graph families. Specifically, 
we ask whether it is possible to color the vertices of a graph with
arboricity $\alpha$ using at most $(2+\eps)\alpha$ many colors 
for any constant $\eps>0$. We answer this question in the affirmative
as well, albeit at the the expense of $O(\eps^{-1}\log n)$ many
passes over the stream.
\begin{itemize}
    \item There is a randomized $O\left(\eps^{-1} \log n\right)$ pass streaming 
        algorithm that finds a $(2+\eps)\alpha$-vertex coloring of the input
        graph using $\tO(n)$ amount of space, for any constant $\eps >0$
        where $\alpha$ is the arboricity of the input graph.
\end{itemize}

\subsection{Our Techniques}

We first discuss a high level overview of 
our $(1+\eps)\Delta$-coloring algorithm.
We observe that, if we do not have any
restriction on the space, then 
there is a simple one pass algorithm for
$(\Delta+1)$-coloring of the input graph.
Assign every vertex the same color initially.
Upon arrival of an edge, store it, and 
if both of its end-points share the same
color, then recolor one of them using an 
``available'' color (which always exists) 
that is not assigned to any of its neighbors. 
However, we do not have the luxury of storing
the entire graph. Hence, we resort to a two 
phase coloring algorithm. In the first phase,
we randomly partition the
vertex set into $O(\Delta / \log n)$ 
many subsets such that for each of these
subsets, the subgraph induced by it has
maximum degree $O(\log n)$ 
with high probability (over the randomness
of the partitioning). The random 
partitioning is realized by coloring each
vertex independently and uniformly
at random using a $O(\Delta / \log n)$ 
size color palette. Then we are able to
store each of these induced subgraphs entirely in
$\tO(n)$ space. Note that the first phase coloring is not a proper 
coloring of the graph $G$, and is carried out
to decompose the graph into smaller monochromatic 
subgraphs. Now, in the second phase 
we color each induced subgraph
with a new palette using the aforementioned
algorithm to get an $O(\Delta)$-vertex coloring of the
original graph. The second phase coloring 
results in a proper vertex coloring of the input graph.
By setting various parameters
suitably, we bound the number of colors 
by $(1+\eps)\Delta$. The detailed description
of the algorithm is given in \cref{sec:deg}.

We now present an overview of our second result, a
$O((1/\eps)\log n)$ pass $(2+\eps)\alpha$-vertex coloring algorithm 
for graphs with arboricity $\alpha$. To design a
color efficient algorithm for bounded arboricity
graphs, we follow the broad strategy employed 
by Barenboim and Elkin in \cite{barenboim2010sublogarithmic}.
They gave a distibuted $O(\alpha)$-vetrex coloring
algorithm that runs in $O(\alpha \log n)$ many rounds.
The crux of their algorithm is to find an orientation
of the edges of the undirected input graph such that
the maximum out-degree of any vertex in the 
oriented directed graph is
bounded by $O(\alpha)$. Call such an orientation
{\em useful}. We show how to find a {\em useful} orientation 
in $O(\log n)$ many passes using $\tO(n)$ space in 
the streaming model, following the template 
of \cite{barenboim2010sublogarithmic}. 
Even with a {\em useful} orientation of 
the edges, we are not quite done. The maximum
out-degree of any vertex is bounded by $O(\alpha)$ in 
the oriented graph, which is prohibitively large with 
the space restriction of the semi streaming model. We employ the
two phase coloring technique discussed above to 
randomly partition the input graph into $O(\alpha/ \log n)$
many induced subgraphs, such that within each 
induced subgraph, the maximum out-degree with respect 
to the {\em useful} orientation is bounded by $O(\log n)$.
By appropriate settings of parameters, we prove that 
this is sufficient to a get a
vertex coloring algorithm that uses at most 
$(2+\eps)\alpha$ many colors for any constant $\eps>0$.
We present this algorithm is details in \cref{sec:arb}.

\subsection{Related Work}

The graph coloring problem is one of the most central
problems in the distributed computing model. 
A monograph by Barenboim and Elkin~\cite{barenboim2013distributed}
gives an excellent overview of the state of the art. 
We mention below few notable results in the synchronous message passing 
({\em SMP}) model. For more detailed results, 
we refer to~\cite{barenboim2013distributed} and
reference therein. In the {\em SMP} model, 
every vertex of the input graph is a processor, and
they communicate with their neighbors (over the edges of the graph)
in synchronous rounds. The running time of an algorithm is the number
of rounds required. There is a randomized $O(\Delta)$-coloring algorithm 
that requires $O\left(\sqrt{\log n}\right)$ time~\cite{kothapalli2006distributed}. 
This result was improved by Schneider and Wattenhofer~\cite{schneider2010new}, and then 
subsequently by Barenboim \etal in \cite{barenboim2016locality} who came up with a 
$(\Delta+1)$-coloring algorithm with running time $2^{O(\sqrt{\log \log n})}$.
Barenboim and Elkin \cite{barenboim2010sublogarithmic} 
studied the $O(\alpha)$-vertex coloring problem for graphs with 
arboricity bounded by $\alpha$. They gave a deterministic 
algorithm that runs in $O(\alpha \log n)$ time and finds 
an $O(\alpha)$-coloring of the input graph, and hence 
remarkably stretching the class of graph families for which
efficient coloring algorithms are known. The main challenge
in converting any of these distributed algorithms
into a streaming algorithm lies in reducing the number 
of rounds. The number of rounds in {\em SMP} model possibly
translates to the number of passes in the streaming settings.
However, in the one pass streaming, it is not clear
how to leverage the distributed algorithms to get 
an efficient streaming algorithm.

The graph coloring problem has been 
studied extensively in the dynamic setting, where the edges of the
graph are inserted and deleted over time, and the goal is to 
maintain a valid vertex coloring of the graph 
after every update. Unlike the streaming setting, in this model
there is no space restriction. The emphasis here 
is to use as few colors as possible while keeping the
{\em update time} small. Bhattacharya \etal~\cite{bhattacharya2018dynamic}
gave a randomized algorithm that maintains $(\Delta + 1)$-vertex 
coloring with $O(\log \Delta)$ expected amortized update time.
They also gave a deterministic algorithm that maintains
$O(\Delta + o(\Delta))$-vertex coloring with $O( \polylog \Delta)$
amortized update time. Barba \etal~\cite{barba2017dynamic}
studied various trade-off between the number of colors used and
update time. However the techniques used in the dynamic settings
do not seem to be readily applicable in the streaming 
setting due to the fundamental differences in the models. 
The problem of edge coloring in the dynamic 
graph has been considered in~\cite{barenboim2017fully,bhattacharya2018dynamic}.
There are many other heuristics based approaches known
in this model with emphasis on experimental supremacy
~\cite{dutot2007decentralized,ouerfelli2011greedy,sallinen2016graph,hardy2018tackling}.

In the streaming model, the problem of coloring an $n$-uniform hypergraph 
using two colors has been studied in~\cite{radhakrishnan2015hypergraph}.
To the best of our knowledge, the general graph coloring 
problem has not been considered in the streaming model before. 

Estimating the arboricity of a graph in the steaming model
is a well studied problem. McGregor \etal~\cite{mcgregor2015densest}
gave a one pass $(1+\eps)$-approximation algorithm to 
estimate the arboricity of graph using $\tO(n)$ space. 
Bahmani \etal~\cite{bahmani2012densest} gave a matching
lower bound.

\section{Preliminaries}

Throughout this paper, we consider the input graph $G=(V,E)$ to 
be a simple undirected graph with $|V|=n$ and $|E|=m$.
We work with streaming model where the input graph
is presented as a stream of edges $(e_1,e_2,\ldots,e_m)$
in some adversarial order. We consider cash-register 
variation of this model in which edges once inserted, are never
deleted. For a vertex $v\in V$, $N(v)$ denotes its set of neighbors 
and $\deg(v) = |N(v)|$ denotes its degree. 
The maximum degree of a graph is denoted by $\Delta = \max_{v\in V} \deg(v)$.

An {\em orientation} of edges of an undirected
graph is an assignment of a direction to each edge of the graph. 
An {\em oriented graph} is a directed graph
obtained by orientation of edges of the corresponding
undirected graph. For a vertex $v$ in a directed graph $G$, 
$N_G^+(v)$ denotes the set of out-neighbors of $v$ and 
$\deg_G^+(v) = |N_G^+(v)|$ denotes its out-degree. We drop
the subscript $G$ if the graph $G$ is clear from the context. 
We denote the maximum out-degree of any vertex by
$\Delta^{\out}$. A directed acyclic graph or DAG is a
directed graph that has no directed cycle.

\begin{definition}[Degeneracy]
A graph $G$ is $\degen$-degenerate if every subgraph $H$ of
$G$ has a vertex of degree at most $\degen$. Degeneracy of a graph is 
the minimum $\degen$ for which the graph is $\degen$-degenerate. 
\end{definition}

\begin{definition}[Arboricity]
Arboricity of an undirected graph is the 
minimum number of forests into which
the edges of the graph can be partitioned. 
We denote the arboricity of $G$ by $\alpha_G$. We drop
the subscript $G$ if the underlying graph $G$ is clear from the context.
By the work of Nash-Williams~\cite{nash1964decomposition}, we have
\[
    \alpha_{G} = \max _{S \subseteq V : |S|>1} \left\lceil \frac{|E(S)|}{|S|-1} \right\rceil \,.
\]
\end{definition}
It is easy to see that both $\degen \leq \Delta$ and $\alpha \leq \Delta$.
A tighter bound of $\alpha \leq \lceil (\Delta+1)/2 \rceil$ is due to 
Chartrand \etal\cite{chartrand1968point}.
It is also known that $\alpha \leq \degen \leq 2\alpha-1$.

\begin{definition}[Vertex Coloring]
A proper vertex coloring of a graph $G=(V,E)$ with a set of colors $C$ 
is a function $\col: V \to C$ such that $\{u,v\}\in E \Rightarrow \col(u)\ne \col(v)$.
\end{definition}

The {\em Chromatic number} of a graph is defined as the minimum number of 
colors needed to get a proper vertex coloring 
of the graph. We denote the chromatic number by $\chi$. 
It trivially holds that $\chi \leq \Delta + 1$. 
For bounded degeneracy graphs, the 
upper bound on $\chi$ improves to $\degen +1$.
Similarly for graphs with arboricity bounded by $\alpha$,
we have $\chi \leq 2\alpha$.

We apply the following versions of the Chernoff bound several times in the paper. 
Let $X_1,\ldots, X_n$ be independent random variables that take values in $\{0, 1\}$. 
Let $X = \sum_{i=1}^n X_i$ denote their sum and 
let $\mu = \EE[X]$ denote its expected value. Then
\begin{align*}
    \Pr[X\geq (1+\delta)\mu)]  
    & \leq \exp\left(-{\frac{\delta^2 \mu}{3}}\right) 
    &\text{ for } 0\leq \delta \leq 1 \,, \\
    \Pr[X\geq (1+\delta)\mu)] 
    &\leq \exp\left(-{\frac{\delta \mu}{3}}\right) 
    & \text{ for } \delta \geq 1\,, \\
    \Pr[X\leq (1-\delta )\mu ] 
    & \leq \exp\left(-{\frac {\delta ^{2}\mu }{3}}\right) 
    & \text{  for } 0\leq \delta \leq 1 \,.
\end{align*}

\section{$(1+\eps)\Delta$-Vertex Coloring}
\label{sec:deg}
In this section we design a simple one pass streaming
algorithm that finds a $(1+\eps)\Delta$-vertex coloring 
of the input graph using $\tO(n)$ amount of space.
We assume that $\Delta$ is known to us. 
In the full version of this paper, we present a slightly
modified algorithm that removes this assumption. 

Suppose $\Delta = O(\log n)$. Hence, we are permitted
to store every edge of the input graph. Then there is a
trivial one pass streaming algorithm that maintains
a $(\Delta+1)$-vertex coloring of the input graph
after every edge update. The algorithm
initializes every vertex with the same color.
Upon arrival of each edge in the stream, it is stored, and 
if both of its end-points share the same
color, then we recolor one of the end points with an 
``available'' color that is not
assigned to any of its neighbors. Since, 
there are $\Delta +1$ colors in the palette,
there always exists an ``available'' color.

Now we consider the interesting case when
$\Delta = \omega(\log n)$. In this case, we perform a two
phase coloring of the input graph. In the first phase, 
we use $\left\lceil{\frac{\eps \Delta}{2c \log n}} \right\rceil$
many colors, where $c$ is a constant to be
set later in the analysis to bound the
probability of success. 
Let $\col: V \rightarrow
\Big\{1,\ldots,{\left\lceil{\frac{\eps \Delta}{2c \log n}} \right\rceil}\Big\}$ denote 
this coloring. Note that this coloring may not be a proper coloring of the graph $G$, 
and is actually done to decompose the graph into smaller induced subgraphs 
(each color induces a subgraph) before the stream arrives. 
The first phase coloring results in
$\left\lceil{\frac{\eps \Delta}{2c \log n}} \right\rceil$ many 
monochromatic subgraphs of $G$. 
Let $V_i$ be the set of vertices with color $i$, 
and $G_i$ be the corresponding
induced subgraph, for all 
$i \in \Big\{1,\ldots,\left\lceil{\frac{\eps \Delta}{2c \log n}} \right\rceil \Big\}$.
Denote by $\Delta_i$ the maximum degree of 
a vertex in the graph $G_i$. In the second phase,
we color each subgraph $G_i$ in parallel, using  
distinct palettes. We show that 
with high probability $\Delta_i \leq (1+2/\eps) c \log n$, 
and hence $(1+2/\eps)c \log n + 1$ many colors are sufficient 
to color each subgraph $G_i$. We, in fact, store 
the entire graph $G_i$ for all $i$ as the stream arrives,
while recoloring the vertices as needed. This can be 
easily done by the algorithm described in the beginning 
of this section for the case $\Delta=O(\log n)$.
Since each $G_i$ is colored using a 
different palette, the second phase coloring 
results in a proper coloring of the original graph $G$. The final coloring of the 
vertices is due to this second phase.
We present the coloring procedure in \cref{algo:color}. 

\begin{algorithm}
    \caption{$(1+\eps)\Delta$-Vertex Coloring Algorithm}
    \label{algo:color}
    \begin{algorithmic}[1] 
        \Require Undirected graph $G=(V,E)$, $\eps$, $\Delta$.
        \Ensure $(1+\eps)\Delta$-vertex coloring of $G$.
        \Statex \textbf{Pre-processing:} 
        \Indent
            \State $\ell \gets \left\lceil{\frac{\eps \Delta}{2c \log n}} \right\rceil $,
                    where $c$ is a large constant. Let $L=\{1,2,\ldots,\ell \}$.
            \State Assign each vertex a color from the set $L$ 
                    independently and uniformly at random. 
            \State Denote this coloring by $\col: V \to L$.
            \State Let $V_i$ be the set of vertices with color $i$, for all $i\in L$.
            \State Let $G_i$ be the subgraph induced by $V_i$, for all $i\in L$.
            \State $r \gets (1+2/\eps)c\log n +1$.
            \State Let $L_i = \{ a_{i,1},a_{i,2},\ldots, a_{i,r} \}$ for all $i\in L$.
            \State Let $\col_i: V_i \to L_i$ for each $i\in L$.
            \State Initialize $\col_i(v)=a_{i,1}$ for all  $i\in L$ and $v \in V_i$.
        \Unindent
        \Statex \textbf {Stream Processing:}
        \Indent 
            \For{each edge $e = \{u,v\}$ in the stream}:
                \If{$\col(u)\ne \col(v)$}: 
                    \State Discard the edge.
                \Else:
                    \State Suppose $u,v\in V_i$.
                    \State Store the edge $e$ in the graph $G_i$.
                    \If{$\col_i(u)  = \col_i(v)$}: \label{ln:recolor}
                        \State Set $\col_i(u)$ to a color in $L_i$ 
                                that is not assigned by $\col_i$ to 
                                any neighbor of $u$. \label{ln:findcolor}
                        \State If such a color is not available, then Abort.
                    \EndIf
                \EndIf
            \EndFor
        \Unindent
    \end{algorithmic}
\end{algorithm}

We now analyze \cref{algo:color}. We first prove that 
with high probability \cref{algo:color} always finds an available 
color in \cref{ln:findcolor} during the execution of the algorithm. 
Since we have $r$ many colors in 
$L_i$, it is sufficient to show that $\Delta_i$ is bounded
by $r-1$ with high probability, over the randomness 
of the first phase coloring. This is handled by \cref{lem:deg_bound}.
Hence, \cref{algo:color}
generates a proper coloring with high probability.

\begin{lemma} 
\label{lem:deg_bound}
Let $\Delta_i$ be the maximum degree of any vertex in the graph $G_i$,
as defined in \cref{algo:color}. Then,
with probability at least $1-{1}/{n^{10}}$, 
$\Delta_i \leq (1+2/\eps) c \log n$ for all $i\in L$.
\end{lemma}

\begin{proof}
Given any $i\in L$, fix a vertex $v\in V_i$. For each neighbor $u$ of $v$ in $G$, 
let $Y_u$ denote the indicator random variable such 
that $Y_u = 1$ if $u$ has the same color as $v$ after the
pre-processing step, 
and $Y_u = 0$ otherwise. Let $X_v =\sum_{u \in N(v)} Y_u$ 
denote the number of neighbors of $v$ with the 
same color as $v$ after the pre-processing step. 
By linearity of expectation, we get
\begin{align*}
    \EE[X_v] = \dfrac{\deg(v)}{\big\lceil{\frac{\eps \Delta}{2c \log n}} \big\rceil} \leq \dfrac{\Delta}{\frac{\eps\Delta}{2c \log n}} = \dfrac{2c}{\eps}\cdot \log n.
\end{align*}
Also, 
\begin{align*}
    \EE[X_v] = \dfrac{\deg(v)}{\big\lceil{\frac{\eps\Delta}{2c \log n}} \big\rceil} 
> \dfrac{1}{\frac{\eps\Delta}{c \log n}} = \dfrac{2c}{\eps}\cdot \dfrac{\log n}{2\Delta}.
\end{align*}
Then we pick some $m \in \{1,\ldots,\Delta\}$ such that 
$\dfrac{2c}{\eps}\cdot\dfrac{\log n}{2m} < E[X_v] \leq \dfrac{2c}{\eps}\cdot\dfrac{\log n}{m}$. 
\begin{align*}
    \Pr[X_v > (1 + 2/\eps)c\log n] 
    & \leq \Pr[X_v > (1+\eps/2)m E[X_v]] 
    && \text{since } E[X_v] \leq \frac{2c}{\eps}\cdot\frac{\log n}{m} \,, \\
    & \leq \exp \left({-\frac{((1+\eps/2)m-1)E[X_v]}{3}}\right)
    && \text{by Chernoff Bound} \,, \\
    & \leq \exp \left({-\frac{(\frac{1+\eps/2}{2}-\frac{1}{2m})\dfrac{2c}{\eps}\cdot\log n}{3}}\right)
    && \text{since } E[X_v] > \dfrac{2c}{\eps}\cdot\frac{\log n}{2m} \,, \\
    & \leq \exp \left({-\dfrac{c\log n}{6}} \right)
    && \text{since } m \geq 1 \,,\\ 
    &= \frac{1}{n^{11}} &&\text{taking }c =\frac{66}{\log e} \,.
\end{align*}
Thus, by union bound, $\Pr[\exists u\in V,~ X_u > (1 + 2/\eps)c\log n] < 1/n^{10}$.
Hence, probability that all vertices have at most $(1 + 2/\eps)c\log n$ 
neighbors with same color 
as themselves is at least $1 - \Pr[\exists u\in V,~  X_u > (1 + 2/\eps)c\log n] > 1- 1/n^{10}$.
\end{proof}

The main result in this section is captured in \cref{thm:color_deg} below.
\begin{theorem}
\label{thm:color_deg}
There is a randomized one pass streaming algorithm that produces a
$(1+\eps)\Delta$-vertex coloring of a graph with maximum degree $\Delta$
using $\tO(n)$ amount of space, for any constant $\eps>0$. Furthermore,
the worst case update time of the algorithm is $O(\eps^{-1}\log n)$.
\end{theorem}
\begin{proof}
The number of colors used by \cref{algo:color} can be upper bounded 
by $\ell \cdot r$. Assuming $\Delta = \omega(\log n)$, we get
\begin{align*}
    \ell \cdot r 
    &=  \Big\lceil{\frac{\eps\Delta}{2c\log n}} \Big\rceil \cdot  
        \left( \left(1+ \frac{2}{\eps} \right)c\log n + 1\right) \,, \\
    &\leq  \left(\frac{\eps\Delta}{2c \log n}+1\right)\left(\left(1+ \frac{2}{\eps} \right)c\log n + 1\right) \,, \\
    &\leq (1 +\eps/2)\Delta + o(\Delta) \,, \\
    & \leq (1 +\eps/2)\Delta + (\eps/2)\Delta = (1+\eps)\Delta \,.
\end{align*}
The space usage of the algorithm is dictated by $\max_{i \in L} \Delta_i$.
From \cref{lem:deg_bound}, we have $\Delta_i = O((1/\eps) \log n)$. Hence,
the \cref{algo:color} requires $O((1/\eps) n \log n)$ amount of 
space. The bound on the update time follows from the 
recoloring time of a vertex in \cref{ln:findcolor}.
\end{proof}

\section{$(2+\eps)\alpha$-Vertex Coloring}
\label{sec:arb}
In this section we discuss a $(2+\eps)\alpha$-vertex coloring algorithm 
in the semi streaming model, where $\alpha$ is the arboricity of the input graph.
This significantly extends the class of graph families for which 
efficient coloring algorithms can be designed. 
In this section, we assume that $\alpha$ is known to the algorithm. 
In the full version of this paper
we discuss how to remove this assumption, albeit at the expense of slightly 
larger palette of colors. 

To design a more color efficient algorithm for 
bounded arboricity graphs, at a high level we follow the 
strategy of Barenboim and Elkin~\cite{barenboim2010sublogarithmic}.
They designed a distributed coloring algorithm with 
$O(\alpha)$ many colors in $O(\alpha \log n)$ many rounds.
We first discuss the central idea of their algorithm, and then
discuss the challenges in implementing those ideas in the streaming model.
Assume, given a graph $G=(V,E$) of arboricity $\alpha$, and a small constant 
$\gamma>0$, we partition the vertices in $V$ into 
$k = O(({1}/{\gamma})\log n )$ many disjoint subsets
$H_1,H_2,\ldots,H_{k}$ such that the following property holds.
\begin{enumerate}[label={P.\arabic*}]
    \item  \textbf{Bounded Degree Vertex Partition:} For every vertex $v\in H_i$, $i\in[k]$, 
            it has at most $(2+\gamma)\alpha$ many neighbours in the vertex set $\cup_{j=i}^{k} H_j$.            
            \label{item:prop}
\end{enumerate}
Such a partitioning then enables us to orient the edges in a way so that 
the resulting directed graph is, in fact, a DAG with maximum out-degree
of any vertex bounded by $(2+\gamma)\alpha$. For instance, consider
the following orientation process. For an edge $\{u,v\}$, orient 
it from the vertex with lower partition number to higher partition number. 
If both $u$ and $v$ are in the same partition, then orient them 
from lower vertex id to higher vertex id. It is not difficult to 
show that such an orientation is acyclic. If there is a 
cycle in the original graph, then there must be at least one vertex
in that cycle that has two outgoing edges in the oriented graph. 
Since any DAG with maximum out-degree $\Delta$ can be colored 
using at most $(\Delta+1)$ many colors, the oriented graph
leads to a $((2+\gamma)\alpha + 1)$-vertex coloring algorithm 
in the distributed settings, although 
not in a straight forward manner. The algorithm requires 
$(\alpha \log n)$ many rounds. Another interesting 
property of the vertex partitioning is that the edge orientations
are implicitly defined by the partition itself. Hence we do not
need to store edge specific information in order to maintain
the oriented graph. We list property of the bounded degree acyclic 
orientation of the edges of a graph in the following item.
\begin{enumerate}[label={P.2}]
    \item  \textbf{Bounded Degree Acyclic Graph Orientation:} 
            Given a graph $G=(V,E)$ with arboricity $\alpha$,
            and a small constant $\gamma >0$, an orientation 
            of the edges is called bounded degree acyclic 
            graph orientation if the orientation is acylic and
            maximum out degree of any vertex given by the 
            orientation is at most $(2+\gamma)\alpha$.
            \label{item:orient}
\end{enumerate}
Note that a vertex partition with property~\ref{item:prop} leads
to an edge orientation with property~\ref{item:orient}.

We now discuss the challenges in converting these ideas into an 
algorithm in the semi streaming model. The first challenge is to 
derive a vertex partitioning with property~\ref{item:prop} using 
only $\tO(n)$ space. This turn out to be a rather easy 
task if we are allowed $O((1 / \gamma)\log n)$ many passes. 
In order to find a vertex partitioning 
that satisfies property~\ref{item:prop}, 
\cite{barenboim2010sublogarithmic} gives a simple greedy
algorithm that iteratively removes 
vertices of degree at most $(2+\gamma)\alpha$ from the graph.
They show after $O( ({1}/{\gamma})\log n )$ many iterations,
the desired partitioning is achieved.
This process easily translates to a $O({1}/{\gamma} \log n)$ pass 
$\tO(n)$ space deterministic algorithm
in the streaming model. For the sake of 
completeness we include a description of this
algorithm in \cref{subsec:orient}. 


The second challenge is to design a one pass streaming 
algorithm that can find a $\Delta^{\out}+1$ coloring
for a DAG with maximum out-degree $\Delta^{\out}$.
Such an algorithm is easy to design in the offline setting, 
where we can store the entire graph. For example, consider a
greedy coloring algorithm that operates on the 
reverse topologically sorted ordering of the 
vertices. It assigns a vertex first available color permitted
by the already colored neighbors of the vertex.
It is easy to see that the algorithm produces a $(\Delta^{\out}+1)$-vertex coloring. 
In the distributed setting, \cite{barenboim2010sublogarithmic}
devises an algorithm that requires $O(\alpha \log n)$
many rounds. We overcome this obstacle 
by leveraging ideas from our $(1+\eps)\Delta$-coloring
algorithm. Instead of working on the DAG directly, 
we consider a two phase coloring process. In the first phase,
we color the vertices using roughly $O(\alpha/ \log n)$ 
many colors. This results in that many monochromatic 
subgraphs, such that each subgraph, when viewed 
as a oriented graph with respect to the vertex partitioning,
has maximum out degree bounded by $O(\log n)$. Hence,
in the second phase, we use the offline algorithm to color each of the
monochromatic subgraphs using a distinct palette. 
Setting the parameters suitably in the big `O' notation, 
we bound the number of colors by $(2+\eps)\alpha$.

\subsection{Graph Orientation}
\label{subsec:orient}
In this section, we give an algorithm to find an 
orientation of the edges that has property~\ref{item:orient}.
Our algorithm is a straight forward adaptation of the distributed 
edge orientation algorithm by Barenboim
and Elkin~\cite{barenboim2010sublogarithmic} in 
the streaming model. The algorithm does not explicitly
orient the edges, rather it finds a partitioning of the 
vertex set that satisfies property~\ref{item:prop}. The
edge orientations are implicitly achieved by this
partitioning. By discussion in the beginning of \cref{sec:arb},
it follows that the orientation has property~\ref{item:orient}.
We now present the procedure in \cref{algo:orient}.
\begin{algorithm}
    \caption{$(2+\gamma)\alpha$-Bounded Degree Graph Orientation Algorithm}
    \label{algo:orient}
    \begin{algorithmic}[1] 
        \Require $G=(V,E)$, $\alpha$, $\gamma$.
        \Ensure A partitioning of the vertex set $V= H_1 \cup H_2 \cup \ldots \cup H_{k}$.
        \Statex \textbf{Vertex Partitoning:}
        \Indent
            \State $G_0 = G, A_0=V$.
            \State $i=1$.
            \While {$A_{i} \neq \emptyset $} \Comment{Each iterations requires one pass over the input stream}
                \State $H_i = \{v\in A_{i-1} ~:~ \deg (v) \text{ in } G_{i-1} \leq (2+\gamma) \alpha \}$.
                \State $A_i = A_{i-1} \setminus H_i$.
                \State Let $G_i$ be the graph induced by $A_i$.
                \State $i = i+1$.
            \EndWhile
            \State Let $k = i-1$. \Comment{$\{H_1,H_2,\ldots,H_k\}$ creates a partition of $V$}
        \Unindent
        \Statex \textbf{Implicit Edge Orientation:}
        \Indent
            \For{each edge $e=\{u,v\} \in G$}
                \State Let with $u \in H_{j_1}$ and $v\in H_{j_2}$, for some $j_1,j_2\in [k]$, 
                        such that $j_1\leq j_2$.  
                \If{$j_1=j_2$}
                    \State Orient $e$ from vertex with lower id to higher id.
                \Else
                    \State Orient $e$ from $u$ to $v$.
                \EndIf
            \EndFor
        \Unindent
    \end{algorithmic}
\end{algorithm}

In analyzing the algorithm, \cite{barenboim2010sublogarithmic} 
showed that $k = O(({1}/{\gamma})\log n)$.
As a result, we have a $O(({1}/{\gamma})\log n)$-pass streaming
algorithm.
\begin{lemma}
\label{lem:orient}
There is a $k$ pass,
$\tO(n)$ space streaming algorithm that partitions the 
vertex set into $k$ many disjoint subsets with property
~\ref{item:prop}, for $k=O\left(\frac{1}{\gamma}\log n \right)$.
\end{lemma}

\subsection{Coloring Algorithm}
\label{subsec:algo}
In this section, we give a $(2+\eps)\alpha$-vertex coloring algorithm. 
Note that if $\alpha = O(\log n)$, then 
we can store the entire graph using $\tO(n)$ space. 
So we consider the interesting case when $\alpha = \omega(\log n)$. 
We assume that $\alpha$ is known to us. 

At first, we use \cref{algo:orient}
to partition the vertex set $V$ into $O((1/ \gamma)\log n)$ many disjoint 
subsets such that property~\ref{item:prop} holds,
for some small constant $\gamma>0$. The parameter 
$\gamma$ is set as a function of the input parameter $\eps$.
This ensures that the maximum out degree of any vertex in 
the implicit orientation of the
edges is bounded by $(2+\gamma)\alpha$. In parallel, we 
consider a decomposition of the input graph into
$O(\alpha/ \log n)$ many subgraphs. This is achieved by 
assigning every vertex a color picked independently and uniformly
at random from a set of $O(\alpha/\log n)$ many colors, and then 
considering monochromatic induced subgraphs. We have already 
demonstrated the usefulness of this idea in designing a 
$(1+\eps)\Delta$-vertex coloring algorithm in \cref{sec:deg}.
Following the same line of argument, we show that in 
every monochromatic induced subgraph, maximum out degree
of any vertex with respect to the edge orientations is
bounded by $O(\log n)$. We give the details of this 
process in \cref{algo:alphacolor}.
\begin{algorithm}
    \caption{$(2+\eps)\alpha$-Vertex Coloring Algorithm}
    \label{algo:alphacolor}
    \begin{algorithmic}[1] 
        \Require $G=(V,E)$, $\alpha$, $\eps$.
        \Ensure $(2+\eps)\alpha$-vertex coloring of $G$.
        \Statex \textbf{Pre-processing:}
        \Indent
            \State $\eps' \gets \eps/6$; $\gamma \gets \eps/3$. 
            \State $\ell \gets \Big\lceil{\dfrac{\eps'}{c}\cdot\dfrac{(2+\gamma)\alpha}{ \log n}} \Big\rceil $,
                    where $c$ is a large constant. Let $L=\{1,2,\ldots,\ell \}$.
            \State Assign each vertex in $G$ a color from the set $L$ 
                    independently and uniformly at random. 
            \State Denote this coloring by $\col: V(G) \to L$.        
            \State Let $G_i$ be the monochromatic subgraph of $G$
            induced by color $i$, for all $i\in L$.
        \Unindent
        \Statex \textbf {Stream Processing:}
        \Indent 
            \For{each edge $e=\{u,v\}$ in the stream}:
            \If{$\col(u) = \col(v)$}: 
                \State Store $e$ in $G$.
            \EndIf
            \EndFor
            \State  In parallel, call \cref{algo:orient} with input $G,\alpha, \gamma$.
                    Denote the output by $\{H_1,H_2,\ldots,H_k\}$.
        \Unindent 
        \Statex \textbf{Post-processing:}
        \Indent
            \State  Let $\Delta_i^{\out}$ be the maximum out-degree in the oriented monochromatic subgraph $G_i$ for all $i\in [\ell]$.
            \State Color $G_i$ with $\Delta_i^{\out}+1$ many new colors using offline coloring algorithm.
        \Unindent
    \end{algorithmic}
\end{algorithm}

We next analyze \cref{algo:alphacolor}. From \cref{lem:orient}
we have, $k = O\left(\frac{1}{\gamma}\log n \right)$. 
We first prove the bound on the number of colors used. 
It is easy to see that the 
algorithm produces a proper coloring, since each subgraph
$G_i$ is colored using a distinct palette. The number of colors is 
bounded by $\sum_{i=1}^{\ell}(\Delta_i^{\out} +1 )$. The space 
usage of the algorithm is $\tO(n \cdot (\max_{i\in [\ell]} \Delta_i^{\out}))$.
Hence, we focus on bounding $\Delta_i^{\out}$, which is handled by
\cref{lem:outdeg_bound}.
\begin{lemma} 
\label{lem:outdeg_bound}
Let $\Delta_i^{\out}$ be as defined in \cref{algo:alphacolor}. Then,
with probability at least $1-\frac{1}{n^{10}}$, $\Delta_i^{\out} \leq (1+1/\eps')c \log n$ 
for all $i\in [\ell]$.
\end{lemma}

\begin{proof}
Given any $i\in L$, fix a vertex $v\in V_i$. 
From the fact that $\{H_1,\ldots,H_k\}$ has property \ref{item:prop} 
and our definition of orientation of the edges, 
it follows that $v$ has out-degree at most $(2+\gamma)\alpha$. For each out-neighbor $u$ of $v$, 
let $Y_u$ denote the indicator random variable such 
that $Y_u = 1$ if $u$ has the same color as $v$ after the
pre-processing step, 
and $Y_u = 0$ otherwise. Then $X_v =\sum_{u \in N^{\out}(v)} Y_u$ denotes the number of out-neighbors of $v$ with the 
same color as $v$ after the pre-processing step.
By linearity of expectation, we get 
\begin{align*}
\EE[X_v] 
= \dfrac{\deg^+(v)}{\Big\lceil{\frac{\eps'}{c}\cdot\frac{(2+\gamma)\alpha}{ \log n}} \Big\rceil} 
\leq \dfrac{(2+\gamma)\alpha}{\frac{\eps'}{c}\cdot\frac{(2+\gamma)\alpha}{ \log n}} 
= \dfrac{c}{\eps'}\cdot \log n.    
\end{align*}
Also,
\begin{align*}
    \EE[X_v] 
= \dfrac{\deg^+(v)}{\Big\lceil{\frac{\eps'}{c}\cdot\frac{(2+\gamma)\alpha}{ \log n}} \Big\rceil} 
> \dfrac{1}{\frac{\eps'}{c}\cdot\frac{2(2+\gamma)\alpha}{ \log n}}
= \dfrac{c}{\eps'}\cdot \dfrac{\log n}{(4+2\gamma)\alpha}.
\end{align*}
Then we can pick $m \in \{1,\ldots,\ceil{(2+\gamma)\alpha}\}$ such that 
$\dfrac{c}{\eps'}\cdot\dfrac{\log n}{2m} < E[X_v] \leq \dfrac{c}{\eps'}\cdot\dfrac{\log n}{m}$.
\begin{align*}
    \Pr[X_v > (1 + 1/\eps')c\log n] 
    & \leq \Pr[X_v > (1+\eps')m E[X_v]] 
    && \text{since } E[X_v] \leq \frac{c}{\eps'}\cdot\frac{\log n}{m} \,, \\
    & \leq \exp \left({-\frac{((1+\eps')m-1)E[X_v]}{3}}\right)
    && \text{by Chernoff Bound} \,, \\
    & < \exp \left({-\frac{\left(\dfrac{1+\eps'}{2}-\dfrac{1}{2m}\right)\dfrac{c}{\eps'}\cdot\log n}{3}}\right)
    && \text{since } E[X_v] > \dfrac{c}{\eps'}\cdot\frac{\log n}{2m} \,, \\
    & \leq \exp \left({-\dfrac{c\log n}{6}} \right)
    && \text{since } m \geq 1 \,,\\ 
    &\leq \frac{1}{n^{11}} &&\text{taking } c \geq \frac{66}{\log e} \,.
\end{align*}
Then, by union bound, we get probability that all vertices have at most $(1+1/\eps')c \log n$ 
out-neighbors with same color as themselves is at least $1 - \Pr[\exists z~ X_z > (1+1/\eps')c \log n]] > 1- 1/n^{10}$.
\end{proof}
\noindent
Thus w.h.p. each oriented monochromatic subgraph $G_i$ will have $\Delta_i^{\out} \leq (1+1/\eps')c \log n$. 
Then total number of colors used after the second phase is
\begin{align*}
    \Big\lceil{\dfrac{\eps'}{c}\cdot\dfrac{(2+\gamma)\alpha}{ \log n}} \Big\rceil \cdot  
        \left( \left(1+ \frac{1}{\eps'} \right)c\log n + 1\right) 
        &\leq  \left(\dfrac{\eps'}{c}\cdot\dfrac{(2+\gamma)\alpha}{ \log n}+1\right)\left(\left(1+ \frac{1}{\eps'} \right)c\log n + 1\right) \,, \\
    &\leq (2+\gamma)(1 +\eps')\alpha + o(\alpha) \,, \\
    & \leq (2 +2\eps/3 + \eps^2/18)\alpha + o(\alpha ) \,, \\
    &\leq (2+\eps)\alpha \,,
\end{align*}
where the second last inequality follows by pluggin in the
values for $\eps'$ and $\gamma$. 
The \cref{thm:color_alpha} below captures our main result.
\begin{theorem}
\label{thm:color_alpha}
Given a graph with arboricity $\alpha$, and a small positive 
constant $\eps$, there is a randomized $O(\frac{1}{\eps}\log n)$
pass streaming algorithm that finds a
$(2+\eps)\alpha$-vertex coloring of the input graph using $\tO(n)$ 
amount of space.
\end{theorem}

\bibliographystyle{alpha}
\bibliography{refs}

\end{document}